\documentclass[superscriptaddress,floatfix,twocolumn]{revtex4}
\usepackage{amsmath, amsthm, amssymb}
\usepackage{graphicx}

\newcommand{\ket}[1]{\left| #1 \right>}

\newtheorem{thm}{Theorem}

\newtheorem{obs}[thm]{Observation}
\newtheorem{lem}[thm]{Lemma}
\newtheorem{dfn}[thm]{Definition}

\newtheorem{cor}[thm]{Corollary}

\begin{document}


\title{3-d topological quantum memory with a power-law energy barrier}	

\author{Kamil P. Michnicki}
\affiliation{Department of Physics, University of Washignton, Seattle, WA USA.}

\begin{abstract}
 We discuss energy barriers and their relationship to self-correcting quantum memories.  We introduce the solid code, a 3-d version of Kitaev's surface code, and then combine several solid codes using a technique called welding. The resulting code is a $[[O(L^3),1,O(L^{\frac{4}{3}})]]$ stabilizer code with an energy barrier of $O(L^{\frac{2}{3}})$, which is an exponential improvement over the previous highest energy barrier in 3-d. No-go results are avoided by breaking microscopic translation invariance.  
\end{abstract}

\maketitle 

\section{Introduction}

An important problem in the quantum computing community is whether it is possible to make a quantum version of the ferromagnetic hard disc drive. Such a medium could be used to protect a quantum state from decoherence without the need to actively detect and correct errors.   The 4-d toric code Hamiltonian, \cite{dennis2002topological, alicki2008thermal}, a spin system, is a theoretical example of such a self-correcting quantum memory.  It uses a macroscopic energy barrier to prevent noise from accumulating and corrupting stored quantum information. It is a major open question whether such a system can exist in less than four dimensions. The problem is intimately related to the problem in condensed matter physics of whether topological order can exist at non-zero temperatures \cite{nussinov2008autocorrelations}. Most results for self-correcting quantum memories in 2-d and 3-d to date have been negative \cite{bravyi2009no,yoshida2011feasibility,landon2013local} or use operators of unbounded strength \cite{hamma2009toric}.   One exception has been the cubic code \cite{haah2011local}, but that too has an energy barrier of only $\log(L)$.  In this letter, we improve the best known energy barrier for spin Hamiltonians with topological order from $O(\log L)$ to $O(L^{2/3})$. Both the Haah code and the code proposed in this letter can be shown \cite{bravyi2011analytic} to give a theoretical increase in storage time when the system size is increased up to a temperature dependent maximum. It is open whether there truly is a maximum system size for these codes past which the storage time decreases but we conjecture that there is. The goal is to have a storage time scaling exponentially with the size, i.e. number of spins, of the system. The result presented in this letter can be viewed as a stepping stone towards this goal.

To gain some intuition, consider the ferromagnetic hard disc drive. It uses the net magnetization of a ferromagnet to store bits of information. At room temperature the net magnetization is stable against a global change in polarization because of a large energy barrier separating states of opposite polarization. If the magnetization of a small domain flips, due to noise, there will be an energy penalty proportional to the perimeter of the domain. At sufficiently low temperatures, this tension tends to shrink it; a phenomenon which leads to a stable classical memory that is self correcting. 

The requirements for quantum memories are more stringent than those for classical memories.  Classical memories need only have a stable order parameter while quantum memories must simultaneously hide the order parameter. Luckily, there are many local error correcting codes that achieve this such as the 2-d and 3-d toric code and color codes \cite{kitaev1997proceedings,castelnovo2007entanglement,bombin2006topological}.  The ground state observables of these depend on the topology of the system, not on any local observables. In this way they hide and therefore protect the stored superposition. Although these systems are resilient against a certain rate of local noise, the noise can still build-up past the point where it can affect observables on the stored quantum state. The 4-d toric code Hamiltonian fights this build up of errors by ensuring that any sequence of local operations that can change such observables must have a macroscopically large energy at some point in the sequence. Thermalization to lower energies prevents such build up of errors. The hope is that if we can find a large energy barrier in 3-d for a local spin system with topological order then the system might be a good quantum memory. 

The toric codes, color codes, cubic code as well as the code presented in this letter are all examples of stabilizer codes. We review stabilizer codes now. The Pauli group is defined by
\begin{equation}
G=\{(i)^k P_1\otimes ... \otimes P_n: k\in\{0,1,2,3 \}, P_i\in \{I,X,Y,Z \}\}.
\end{equation}
where $X$, $Y$ and $Z$ are single qubit Pauli operators. A stabilizer group $S$ is a subgroup of the Pauli group such that $-I \notin S$. This implies that there exists a subspace $\mathcal{H}_c$, such that for all $\ket{\psi}\in \mathcal{H}_c$ and for all $h\in S$, $h\ket{\psi}=\ket{\psi}$. This subspace is what we call the code space. The protocol to do error correction with a stabilizer code is to measure operators from the stabilizer group and to perform error correcting operations based on the value of the measurements. For an ideal self-correcting quantum memory, the environment does the error correction by thermalizing to lower energy states. 

Logical operators are operators that map one codeword to another and keep correctable states correctable.  For simplicity we will choose logical operators that are in the Pauli group and which commute with the stabilizer group. The code presented in this letter is of a special type \cite{calderbank1996good}. It has a generating set where each generator is either a tensor product of exclusively $X$ operators or exclusively $Z$ operators. We will be considering logical operators of the same type. 

The energy barrier of a code is defined with respect to a Hamiltonian. The Hamiltonian is defined in terms of a generating set $R$ such that the group generated by multiplication of elements of $R$ generates $S$. That is, $\langle R \rangle = S$. The Hamiltonian is $H=-\sum_{h \in R} h$.  The ground state subspace is exactly the code space of the stabilizer group $S$. The energy barrier is defined with respect to a local sequence of errors that maps one ground state to another, i.e. that enacts a logical operator. The minimum peak energy over all such sequences is what we call the energy barrier.

In this letter, we will be discussing the energy barrier exclusively. But the motivation to study the energy barrier comes from considering the lifetime of a memory subject to thermal noise. Thermal noise is typically modeled \cite{bravyi2011analytic, chesi2010thermodynamic, alicki2008thermal, davies1976quantum} by a set of local jump operations, i.e. the errors, which occur with a rate proportional to $\exp(-\beta \Delta)$. $\Delta$ is the change in energy upon the application of the error. Thus the larger the energy barrier the more thermal noise is suppressed in this model. 

\textbf{Main result}--There exists a local stabilizer Hamiltonian with an energy barrier of $O(L^{\frac{2}{3}})$ where the Hamiltonian is composed of $O(L^3)$ qubits and the qubits are of finite density. By finite density we mean that a finite number of qubits fit into a finite volume reference box. By local we mean that the terms in the Hamiltonian act on a set of qubits contained in another finite reference box.

Haah \cite{haah2013commuting} proved that for local translation-invariant stabilizer codes, the highest energy barrier is $O(\log L)$. This improves a no-go theorem by Yoshida \cite{yoshida2011feasibility}. So how exactly can the result in this letter hold? The code in this letter is constructed from macroscopic blocks. Each such block satisfies the result by Haah. We describe these blocks in section IIA. These macroscopic blocks are joined together, welded, into a macroscopic lattice. This welding process is described in section IIB. The code in this letter is translation invariant over a length that grows with the system size. This avoids the no-go results by Yoshida \cite{yoshida2011feasibility} and Haah \cite{haah2013commuting}.

\section{Solid Codes and Welding}
\subsection{Solid Codes}

In this section we introduce the solid code which is a local stabilizer code on a three dimensional lattice of qubits. We use the word qubit instead of spin to emphasize that the subsystems of the Hamiltonian could be any two level system.  The solid code is a stabilizer code. It is the 3-d analog of a surface code \cite{dennis2002topological}, i.e. a 3-d toric code \cite{castelnovo2007entanglement} with rough and smooth boundaries. We discuss its logical operators and show that one of them has a constant energy barrier which means it cannot store quantum information though it is an important building block of the welded codes.

We define the generators of the solid code with respect to a graph, shown in figure \ref{solid} (qubits are labeled by edges). The graph is a cubic lattice, i.e. a cube composed of $d \times d \times d$ cubic primitive cells but with the horizontal edges removed for primitive cells at the top and bottom boundaries. These qubits are not included in our code. Terms in the Hamiltonian are labeled by vertices and faces of this graph, where by faces we mean faces of the primitive cells. These faces will be referred to as plaquettes. The vertices are:
 \begin{equation}
V=\{v=(v_1,v_2,v_3):v_i \in \{1,...,N\}\}.
\end{equation}
Using the unit vectors $n_1=(1,0,0)$, $n_2=(0,1,0)$ and $n_3=(0,0,1)$, the edges are:
\begin{equation}
\begin{split}
E = \{ \{v,v+n_3\}:v\in V, v_3\neq N \} \\ 
\cup \{ \{v,v+n_2\}:v\in V, v_3\neq 1, v_3\neq N \} \\
\cup \{ \{v,v+n_1\}:v\in V,v_3\neq1,v_3\neq N \}.
\end{split}
\end{equation}
Let $\Gamma(v)$ be the set of edges that neighbor a vertex $v$. For each $v$ with $|\Gamma(v)|>1$ define the term $h_v^X=\prod_{e\in \Gamma(v)} X_e $. Let $\partial f$ be the set of edges on the boundary of a plaquette $f$.  For each plaquette $f$ define a term $h_f^Z=\prod_{e\in \partial f} Z_e $.  The plaquettes along the top and bottom rough boundaries are missing one edge. Those qubits aren't included in our code. $|\partial f|=3$ there. Finally the Hamiltonian is given by a sum over the elements of the set of vertices $V$ and the set of faces $F$:
\begin{equation}
H=\left( \sum_{v \in V:|\Gamma(v)|>1} -h_v^X \right) +\left( \sum_{f \in F} -h_f^Z \right).
\end{equation}
The $X$ and $Z$-type terms overlap on an even number of qubits and so commute. Hence the terms generate a stabilizer group and the ground state subspace of the Hamiltonian is exactly the code space of this stabilizer group.

The set of logical operators completely determines the ground state degeneracy. We will show that there are two distinct non-trivial logical operator. We will denote them by $\bar{X}$ and $\bar{Z}$.  Because they anti-commute, we can only diagonalize one of them at a time. Hence the ground state degeneracy is 2. 

The logical operator $\bar{X}$ of the solid code resembles an open membrane. If $\ket{\psi}$ is in the ground subspace of $H$ then each vertex operator satisfies $h^X_v \ket{\psi}=\ket{\psi}$. Similarly a product of vertex operators has a $+1$ eigenvalue. Multiplying vertex operators generates closed membranes of qubits and horizontal pairs of open membranes. These open membranes can be made to be far apart so that each membrane overlaps with a disjoint set of plaquettes. Hence each membrane commutes with each term in the Hamiltonian, yet is not generated by them. We conclude that the logical operator $\bar{X}$ is a tensor product of single qubit $X$  operators on a single horizontal membrane. 

The logical operator $\bar{Z}$ resembles an open string. If $\ket{\psi}$ is in the ground subspace then a plaquette operator satisfies $h^Z_f\ket{\psi}=\ket{\psi}$. Similarly a product of plaquette operators has a $+1$ eigenvalue. Multiplying plaquette operators generates closed strings, strings starting and ending on the same rough boundary and pairs of strings extending between opposite rough boundaries. These pairs of open strings can be made to be far apart so that each string overlaps with a disjoint set of vertex operators. Hence each open string commutes with each term of the Hamiltonian, yet is not generated by them. We conclude that the logical operator $\bar{Z}$ is a tensor product of single qubit $Z$ operators on a single string extending between opposite rough boundaries. See figure \ref{solid}.

The logical operator $\bar{Z}$ for the solid code has a constant energy barrier. Understanding why is key to doing better. Applying a $Z$ operator on a single qubit violates either one or two terms. We call these violated terms defects. By flipping an adjacent qubit, we satisfy that term while violating at most one other term. In this way we can move a defect in the bulk of the solid or annihilate it at either of the rough boundaries. By flipping a qubit on a rough boundary, we create a single defect which we can then move to and annihilate on the opposite boundary. This sequence has no more than a constant energy penalty. 

The sequence of errors resulting in the logical operator $\bar{Z}$ resembles a growing string.  One way to create a large energy barrier is to force this string to split many times. To do this we need qubits such that errors on them create three or more defects.  We will be combining blocks of solid code in the next section to achieve this.

\subsection{Welded Solid Codes}

In this section we achieve a power-law energy barrier by combining several solid codes into a 3-d lattice. It is interesting to note that the final lattice is not a regular lattice. This is because each block is bent and stretched to match-up and connect with each other into a macroscopic lattice. The procedure for combining blocks of code is called welding. We will weld three solid codes together,  analyze the shape of the logical operators and show that the energy barrier has increased. 

To gain some intuition, consider a 1-d Ising model of a finite length. Here the Hamiltonian is $\sum_{i=0}^{n-1} -Z_i Z_{i+1}$. The ground state has two degenerate eigenstates $\ket{00...0}$ and $\ket{11...1}$. Suppose we want to flip all of the qubits in a sequence that minimizes the number of defects. The best we can do is to flip the first, second, third etc... qubits in a line until all of the qubits have been flipped. This sequence creates a single defect and moves it from one end of the string to the other. Welding is like combining three such strings on the last qubit, i.e. so each string shares the $n$th qubit with each other. When we try to move a defect past this shared qubit, it will split into two. Thus the energy increases. For the welded solid code there is 2-d boundary between 3-d blocks of qubits. Defects split when moving past these boundaries. But now there is a choice about which part of the boundary we move the defect through since it is an area not a point. 

We now describe this boundary where defects split by combining three solid codes along their rough boundaries. For each solid we identify qubits on the bottom rough boundaries with each other. More precisely, for all $i,j\in \{1,...,N\}$ the qubit labeled by $(\{i,j,0\},\{i,j,1\})$ in the first solid code is the same as the corresponding qubit in the second and third solid codes.  Because the $X$ and $Z$-type terms no longer commute, we update all local $Z$-type stabilizers to commute with the $X$-type stabilizers. The result is that whenever $Z$-type stabilizers agree on the shared qubits of the three solids, they are combined into a single operator. Combining operators is what we call welding. We can formalize welding as follows. For a $Z$-type operator $h$, define $Q(h)$ to be the qubits that $h$ acts on nontrivially. A set of $Z$-type operators $\{h_1,...,h_n\}$ is said to be {\em welded } together into a $Z$-type operator $h$ when $Q(h)=\cup_i Q(h_i)$. If we treat the identity operator $I$ as a $Z$-type operator, then every $Z$-type stabilizer of the new code is a welded version of $Z$-type stabilizers from the three solid codes. Thus we call the resulting code a welded code.  A more thorough account of the theory of welding can be found in the supplementary material and \cite{michnicki20123}.

Next, we show that the new code, the three welded solids, encodes only a single qubit by showing that all nontrivial $Z$-type logical operators are equivalent. The $Z$-type logical operators from each solid get welded together so that the new logical operator $\bar{Z}$ resembles three strings emanating from a single qubit on the shared boundary. See figure \ref{weldedsolidcode}.  After welding, the $Z$-type stabilizers have the following shapes: half-loops on rough boundaries, loops in the bulk of each solid, three welded half-loops on the shared rough boundary and pairs of logical Z operators.  There can be no other $Z$-type logical operator. This is because if a $Z$-type operator acts with an even number of $Z$ operators on the shared qubits, then it is in the stabilizer group and if it acts with an odd number of $Z$ operators on the boundary, then it is equivalent the operator composed of three welded strings. Since there is only one nontrivial $Z$-type logical operator, it follows that the ground state degeneracy is 2, corresponding to a single logical qubit. 

A defect caused by $Z$ errors would split into two moving through this shared boundary, an increase in energy. This is because a single-qubit $Z$ operator applied to this shared boundary creates three violated terms, one for each solid block.  Thus the energy barrier for the logical operator $\bar{Z}$ has increased from 1 to 2.

\begin{figure}[h!]
  \includegraphics[width=0.5\textwidth]{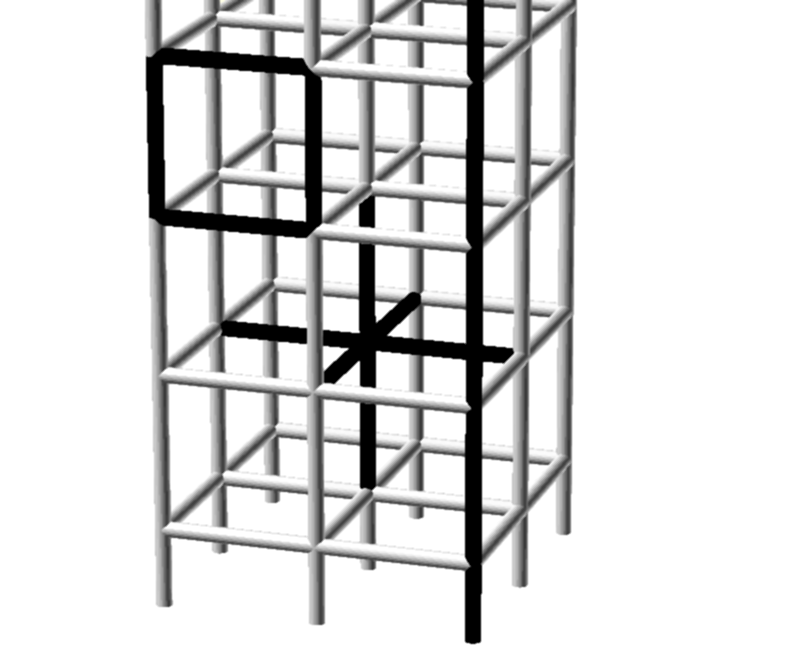}
  \caption{A solid code with qubits represented by edges. The following operators shaded darker: a plaquette operator, a star operator and a logical Z operator.}
  \label{solid}
\end{figure}

\begin{figure}[h!]
  \includegraphics[width=0.5\textwidth]{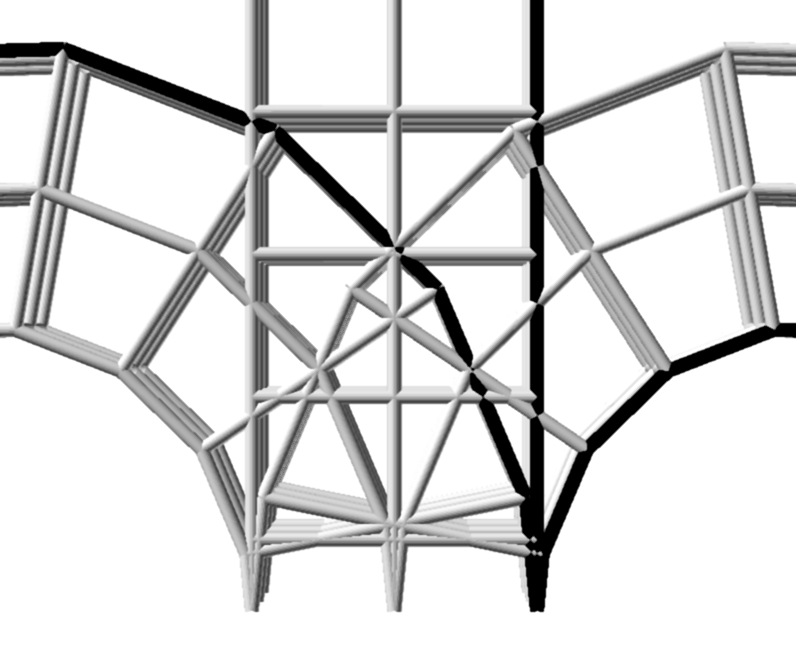}
  \caption{Three solid codes welded together with qubits represented by edges. The bifurcating $\bar{Z}$ operator is shaded darker.}
  \label{weldedsolidcode}
\end{figure} 

In order to increase the energy barrier to a power law, we generalize welding three solids on a single boundary to welding many solids into a lattice. We label the solids by their rough boundaries, where a rough boundary is denoted by $W_i$. In the previous example we had the set of solids $E=\{\{W_1,W_4\},\{W_2,W_4\},\{W_3,W_4\} \}$ and they all share the rough boundary $W_4$, where defects split. In this notation, the welded rough boundaries act as``fat'' vertices and the bulk of the solid codes act as ``fat'' edges of a graph $G=(V=\{W_i\},E)$.  We weld the solid codes into the graph of a 3-d cubic lattice.  As we will see, this gives a particularly high energy barrier.

We deduce the energy barrier of the logical operator $\bar{Z}$ of the welded cubic lattice by looking at the rough boundaries of each block.  If there is an odd number of $Z$ errors on the rough boundaries of a particular solid code, then there must be at least one defect in the bulk of that solid. Hence, if we weld the solid codes into a graph $G$, then the energy barrier of the welded string operator will be at least as big as the energy barrier for an Ising model Hamiltonian $H=\sum_{\{i,j\}\in E} -Z_i Z_j$ with precisely the same graph $G$. If $G$ is a 3-d cubic lattice of width $R$, then the energy barrier for the logical operator $\bar{Z}$ will be $O(R^2)$. In fact, this bound can be saturated provided we never create more than one defect within the bulk of any solid.

We can deduce the energy barrier of the logical operator $\bar{X}$ in a similar way as for the logical operator $\bar{Z}$. The logical operator $\bar{X}$ is an $X$-type operator and hence does not get welded. It remains a membrane. We show a lower bound for the energy barrier by considering the contribution to the energy barrier from the vertical plaquettes only, i.e. plaquettes in the $y$-$z$ and $z$-$x$ planes of each solid, leaving out the plaquettes in the $x$-$y$ plane. Notice that for each vertical surface of plaquettes, defects can move up and down without creating new vertical defects but moving between these vertical regions creates a vertical defect in each neighboring region, synonymously to the case of the solid regions. These ``flat'' regions are connected to each other in a 2-d square lattice of width $O(d)$, provided each solid is $O(d)$ qubits wide. Again the energy barrier of this horizontal membrane will be given by the energy barrier of the Ising model on a 2-d square lattice.  So the energy barrier is lower bounded by $O(d)$. This bound can be saturated, even if we include horizontal plaquettes, provided that the membrane is grown completely horizontally and in a single domain. 

Finally, the energy barrier of solid codes welded in a cubic lattice is the minimum of the two energy barriers: $O(d)$ and $O(R^2)$. The total number of qubits will be $O(d^3)$ qubits per solid with $O(R^3)$ solids, which leads to the number of qubits $N\sim O(d^3 R^3)$. The maximum energy barrier for a fixed number of qubits $N$ is the minimum of the $X$ and $Z$ energy barriers. Thus the maximum energy barrier happens when $O(d)\sim O(R^2)$, leading to an energy barrier of $\delta E\sim O(N^{2/9})$. The qubits can be placed in a box of side lengths of $O(L)$ so that the energy barrier is $O(L^{2/3})$. This demonstrates our main result.

\section{Discussion}

We have constructed a code that has an exponentially higher energy barrier than the the logarithmic bound presented by Haah \cite{haah2013commuting}. We achieved this by tuning the length over which the code is periodic to a macroscopic distance. 

The time that it takes for a memory to be corrupted depends not only on the energy barrier but also on the number of error sequences that lead to a logical operator.  \cite{bravyi2011analytic} derived a lower bound on the storage time $t$ for any stabilizer code Hamiltonian with energy barrier $\delta E$ of $t \sim \frac{e^{\beta \delta E}}{N}2^{-k(L)}$ when $N=O(L^3)\lesssim e^{\beta}$ using the error model of a Hamiltonian in the weak coupling limit. This leads to a lower bound of $t \sim e^{\beta e^{\frac{2}{9}\beta}}$ for $N \lesssim e^{\beta}$ for the welded solid code. An upper bound is not known. 

Further progress might be made by considering non-periodic codes or local dynamic models such as cellular automata decoders.

{\bf Acknowledgements.}
Thank you Aram Harrow and Steve Flammia for stimulating discussion and helpful suggestions throughout the writing process, and to Jeongwan Haah for feedback on the final manuscript. This work was supported by NSF grant 0829937, DARPA QuEST contract FA9550-09-1-0044 and IARPA via DoI NBC contract D11PC20167. Part of this work was carried out while visiting at MIT.

\bibliography{refs}

\begin{thebibliography}{17}
\expandafter\ifx\csname natexlab\endcsname\relax\def\natexlab#1{#1}\fi
\expandafter\ifx\csname bibnamefont\endcsname\relax
  \def\bibnamefont#1{#1}\fi
\expandafter\ifx\csname bibfnamefont\endcsname\relax
  \def\bibfnamefont#1{#1}\fi
\expandafter\ifx\csname citenamefont\endcsname\relax
  \def\citenamefont#1{#1}\fi
\expandafter\ifx\csname url\endcsname\relax
  \def\url#1{\texttt{#1}}\fi
\expandafter\ifx\csname urlprefix\endcsname\relax\def\urlprefix{URL }\fi
\providecommand{\bibinfo}[2]{#2}
\providecommand{\eprint}[2][]{\url{#2}}

\bibitem[{\citenamefont{Dennis et~al.}(2002)\citenamefont{Dennis, Kitaev,
  Landahl, and Preskill}}]{dennis2002topological}
\bibinfo{author}{\bibfnamefont{E.}~\bibnamefont{Dennis}},
  \bibinfo{author}{\bibfnamefont{A.}~\bibnamefont{Kitaev}},
  \bibinfo{author}{\bibfnamefont{A.}~\bibnamefont{Landahl}}, \bibnamefont{and}
  \bibinfo{author}{\bibfnamefont{J.}~\bibnamefont{Preskill}},
  \bibinfo{journal}{Journal of Mathematical Physics}
  \textbf{\bibinfo{volume}{43}}, \bibinfo{pages}{4452} (\bibinfo{year}{2002}),
  \bibinfo{note}{arXiv:quant-ph/0110143v1}.

\bibitem[{\citenamefont{Alicki et~al.}(2008)\citenamefont{Alicki, Horodecki,
  Horodecki, and Horodecki}}]{alicki2008thermal}
\bibinfo{author}{\bibfnamefont{R.}~\bibnamefont{Alicki}},
  \bibinfo{author}{\bibfnamefont{M.}~\bibnamefont{Horodecki}},
  \bibinfo{author}{\bibfnamefont{P.}~\bibnamefont{Horodecki}},
  \bibnamefont{and}
  \bibinfo{author}{\bibfnamefont{R.}~\bibnamefont{Horodecki}},
  \bibinfo{journal}{Arxiv preprint arXiv:0811.0033}  (\bibinfo{year}{2008}).

\bibitem[{\citenamefont{Nussinov and
  Ortiz}(2008)}]{nussinov2008autocorrelations}
\bibinfo{author}{\bibfnamefont{Z.}~\bibnamefont{Nussinov}} \bibnamefont{and}
  \bibinfo{author}{\bibfnamefont{G.}~\bibnamefont{Ortiz}},
  \bibinfo{journal}{Physical Review B} \textbf{\bibinfo{volume}{77}},
  \bibinfo{pages}{064302} (\bibinfo{year}{2008}).

\bibitem[{\citenamefont{Bravyi and Terhal}(2009)}]{bravyi2009no}
\bibinfo{author}{\bibfnamefont{S.}~\bibnamefont{Bravyi}} \bibnamefont{and}
  \bibinfo{author}{\bibfnamefont{B.}~\bibnamefont{Terhal}},
  \bibinfo{journal}{New Journal of Physics} \textbf{\bibinfo{volume}{11}},
  \bibinfo{pages}{043029} (\bibinfo{year}{2009}),
  \bibinfo{note}{arXiv:0810.1983v2}.

\bibitem[{\citenamefont{Yoshida}(2011)}]{yoshida2011feasibility}
\bibinfo{author}{\bibfnamefont{B.}~\bibnamefont{Yoshida}},
  \bibinfo{journal}{Annals of Physics}  (\bibinfo{year}{2011}),
  \bibinfo{note}{arXiv:1103.1885v3}.

\bibitem[{\citenamefont{Landon-Cardinal and Poulin}(2013)}]{landon2013local}
\bibinfo{author}{\bibfnamefont{O.}~\bibnamefont{Landon-Cardinal}}
  \bibnamefont{and} \bibinfo{author}{\bibfnamefont{D.}~\bibnamefont{Poulin}},
  \bibinfo{journal}{Physical review letters} \textbf{\bibinfo{volume}{110}},
  \bibinfo{pages}{090502} (\bibinfo{year}{2013}).

\bibitem[{\citenamefont{Hamma et~al.}(2009)\citenamefont{Hamma, Castelnovo, and
  Chamon}}]{hamma2009toric}
\bibinfo{author}{\bibfnamefont{A.}~\bibnamefont{Hamma}},
  \bibinfo{author}{\bibfnamefont{C.}~\bibnamefont{Castelnovo}},
  \bibnamefont{and} \bibinfo{author}{\bibfnamefont{C.}~\bibnamefont{Chamon}},
  \bibinfo{journal}{Physical Review B} \textbf{\bibinfo{volume}{79}},
  \bibinfo{pages}{245122} (\bibinfo{year}{2009}).

\bibitem[{\citenamefont{Haah}(2011)}]{haah2011local}
\bibinfo{author}{\bibfnamefont{J.}~\bibnamefont{Haah}},
  \bibinfo{journal}{Physical Review A} \textbf{\bibinfo{volume}{83}},
  \bibinfo{pages}{042330} (\bibinfo{year}{2011}),
  \bibinfo{note}{arXiv:1101.1962v2}.

\bibitem[{\citenamefont{Bravyi and Haah}(2011)}]{bravyi2011analytic}
\bibinfo{author}{\bibfnamefont{S.}~\bibnamefont{Bravyi}} \bibnamefont{and}
  \bibinfo{author}{\bibfnamefont{J.}~\bibnamefont{Haah}},
  \bibinfo{journal}{Arxiv preprint arXiv:1112.3252}  (\bibinfo{year}{2011}).

\bibitem[{\citenamefont{Kitaev et~al.}(1997)\citenamefont{Kitaev, Hirota,
  Holevo, and Caves}}]{kitaev1997proceedings}
\bibinfo{author}{\bibfnamefont{A.}~\bibnamefont{Kitaev}},
  \bibinfo{author}{\bibfnamefont{O.}~\bibnamefont{Hirota}},
  \bibinfo{author}{\bibfnamefont{A.}~\bibnamefont{Holevo}}, \bibnamefont{and}
  \bibinfo{author}{\bibfnamefont{C.}~\bibnamefont{Caves}}
  (\bibinfo{year}{1997}).

\bibitem[{\citenamefont{Castelnovo and
  Chamon}(2007)}]{castelnovo2007entanglement}
\bibinfo{author}{\bibfnamefont{C.}~\bibnamefont{Castelnovo}} \bibnamefont{and}
  \bibinfo{author}{\bibfnamefont{C.}~\bibnamefont{Chamon}},
  \bibinfo{journal}{Physical Review B} \textbf{\bibinfo{volume}{76}},
  \bibinfo{pages}{184442} (\bibinfo{year}{2007}),
  \bibinfo{note}{arXiv:0804.3591v2}.

\bibitem[{\citenamefont{Bombin and
  Martin-Delgado}(2006)}]{bombin2006topological}
\bibinfo{author}{\bibfnamefont{H.}~\bibnamefont{Bombin}} \bibnamefont{and}
  \bibinfo{author}{\bibfnamefont{M.~A.} \bibnamefont{Martin-Delgado}},
  \bibinfo{journal}{arXiv preprint quant-ph/0605138}  (\bibinfo{year}{2006}).

\bibitem[{\citenamefont{Calderbank and Shor}(1996)}]{calderbank1996good}
\bibinfo{author}{\bibfnamefont{A.}~\bibnamefont{Calderbank}} \bibnamefont{and}
  \bibinfo{author}{\bibfnamefont{P.}~\bibnamefont{Shor}},
  \bibinfo{journal}{Physical Review A} \textbf{\bibinfo{volume}{54}},
  \bibinfo{pages}{1098} (\bibinfo{year}{1996}),
  \bibinfo{note}{arXiv:quant-ph/9512032v2}.

\bibitem[{\citenamefont{Chesi et~al.}(2010)\citenamefont{Chesi, Loss, Bravyi,
  and Terhal}}]{chesi2010thermodynamic}
\bibinfo{author}{\bibfnamefont{S.}~\bibnamefont{Chesi}},
  \bibinfo{author}{\bibfnamefont{D.}~\bibnamefont{Loss}},
  \bibinfo{author}{\bibfnamefont{S.}~\bibnamefont{Bravyi}}, \bibnamefont{and}
  \bibinfo{author}{\bibfnamefont{B.~M.} \bibnamefont{Terhal}},
  \bibinfo{journal}{New Journal of Physics} \textbf{\bibinfo{volume}{12}},
  \bibinfo{pages}{025013} (\bibinfo{year}{2010}).

\bibitem[{\citenamefont{Davies}(1976)}]{davies1976quantum}
\bibinfo{author}{\bibfnamefont{E.~B.} \bibnamefont{Davies}}
  (\bibinfo{year}{1976}).

\bibitem[{\citenamefont{Haah}(2013)}]{haah2013commuting}
\bibinfo{author}{\bibfnamefont{J.}~\bibnamefont{Haah}},
  \bibinfo{journal}{Communications in Mathematical Physics}
  \textbf{\bibinfo{volume}{324}}, \bibinfo{pages}{351} (\bibinfo{year}{2013}).

\bibitem[{\citenamefont{Michnicki}(2012)}]{michnicki20123}
\bibinfo{author}{\bibfnamefont{K.}~\bibnamefont{Michnicki}},
  \bibinfo{journal}{arXiv preprint arXiv:1208.3496}  (\bibinfo{year}{2012}).

\end{thebibliography}

\section{Appendix:the Theory of Welding}

 Welding is a technique for combining two stabilizer groups to produce a third. It can be used to combine the shape of their logical operators while keeping the generating set local. 

The motivation for the technique of welding comes from the problem of combining two CSS stabilizer groups $S_1$ and $S_2$. Since they are of the CSS type, we have the identities $S_1 = \left<S_1^X \cup S_1^Z \right>$ and $S_2 = \left<S_2^X \cup S_2^Z \right>$ where $S_i^X$ (resp. $ S_i^Z$) refers to the subgroup of $X$-type (resp. $Z$-type) operators of the stabilizer group $S_i$. The stabilizer group $S_1$ is defined on qubits $Q_1$ and the stabilizer group $S_2$ is defined on qubits $Q_2$. When the condition $Q_1 \cap Q_2 \neq \emptyset$ is satisfied then the group $\left< S_1 \cup S_2 \right>$ is not necessarily Abelian and thus not necessarily a stabilizer group. This is because there may be anticommuting pairs between the sets $S_1^X \cup S_1^Z$ and $S_2^X \cup S_2^Z$. One way to get around this problem is to choose the $X$-type subgroup to be $\left< S_1^X \cup S_2^X \right>$ and then update the subgroups $S_1^Z$ and $S_2^Z$ to commute with these operators.  This new code is what we call a welded code. 

When a $Z$-type operator $h$ commutes with each element of the set $S_1^X \cup S_2^X$, it takes on a special form. When restricted to qubits $Q_1$(or $Q_2$), $h$ is a $Z$-type stabilizer or nontrivial logical operator of the stabilizer group $S_1$(or $S_2$). That is, $h$ contains either $Z$-type stabilizers or logical operators of the stabilizer groups $S_1$ and $S_2$ as substrings.

Generating sets can be kept local after welding provided that they satisfy certain conditions. These conditions are called \textit{well matched} and \textit{independent on the weld}.  Furthermore, the nontrivial logical operators of the stabilizer group $S_3$ always look like nontrivial logical operators of the stabilizer groups $S_1$ and $S_2$ when restricted to qubits $Q_1$ or $Q_2$. 

Next we make a series of definitions and observations that make these introductory remarks more clear. Throughout these definitions $i\in \{ 1,2,3 \}$ where we are welding code 1 and code 2 into code 3. 

\begin{dfn}
The group $S_i$ denotes a stabilizer group of the CSS type, i.e. $S_i=\left< S_i^X \cup S_i^Z \right>$ for $S_i^X$ and $S_i^Z$ defined below. 
\end{dfn}

\begin{dfn}
The group $S_i^X$ contains all $X$-type stabilizers of the group $S_i$, i.e. $S_i^X=S_i\cap \left< X_1, X_2,...,X_n \right>$.
\end{dfn}

\begin{dfn}
The group $S_i^Z$ contains all $Z$-type stabilizers of the group $S_i$, i.e. $S_i^Z=S_i\cap \left< Z_1, Z_2,...,Z_n \right>$.
\end{dfn}

Here, $X_k$ (or $Z_k$) denote a bit flip (or phase flip) on qubit $k$. Without loss of generality, we will assume that $X$ and $Z$-type stabilizers and logical operators have $+1$ coefficients. Generally, an $X$(or $Z$)-type stabilizer can have a $+1$ or $-1$ coefficient. As an example, $-X_1 X_2 X_3$ is a valid stabilizer with a $-1$ coefficient.  However, the logical operators do not depend on the choice of the eigenvalues of the stabilizers.

\begin{dfn}
The set $H_i^X$ denotes a generating set of the group $S_i^X$, i.e. $\left< H_i^X \right>=S_i^X$.
\end{dfn}

\begin{dfn}
The set $H_i^Z$ denotes a generating set of the group $S_i^Z$, i.e. $\left< H_i^Z \right>=S_i^Z$.
\end{dfn}

\begin{dfn}
The group $N_i$ denotes the normalizer of the group $S_i$, i.e. $h\in N_i$ if and only if $h\in \left<X_1,...,X_n,Z_1,...,Z_n \right>$ and for all $g\in S_i$, $[h,g]=0$.
\end{dfn}


\begin{dfn}
The group $N_i^Z$ contains all $Z$-type operators that commute with every element of $S_i$, i.e. $N_i^Z=N_i\cap\left< Z_1,...,Z_N \right>$.
\end{dfn}


\begin{dfn}
The set $L_i^Z $ is a minimal generating set for all $Z$-type logical operators, i.e. it includes all $Z$-type logical operators means that $\left< L_i^Z \cup S_i^Z \right> = N_i^Z$ and it is minimal means that  $\left< L_i^Z \right> \cap S_i^Z=\{I \}$.
\end{dfn}

The sets $L_i^Z$, $H_i^X$ and $H_i^Z$ are assumed to be uniquely defined in the observations and theorems of this appendix. But any choice satisfying  the conditions of the observations and theorems will do. 

\begin{dfn}
The set $Q(S)$ denotes qubits that a stabilizer code $S$ acts on non-trivially, i.e. qubit $q\in Q(S)$ iff there exists $h\in S$ such that $tr_q h=0$. To simplify notation we denote $Q(S_i)$ as $Q_i$. 
\end{dfn}

\begin{dfn}
The operator $\theta_{Q}(O)$ denotes the restriction of a product operator $O$ to qubits $Q$.  More specifically, if $O=\otimes_i P_i$ then $\theta_Q (O)=\otimes_{i\in Q} P_i \otimes_{j\notin Q} I_j$.  When applied to a set of product operators $S$, $Q(S)=\{\theta_Q(h): h \in S \}$. We simplify the notation by defining $\theta_{Q_1} \equiv \theta_1$, $\theta_{Q_2} \equiv \theta_2$ and $\theta_{Q_1\cap Q_2} \equiv \theta_{12}$.
\end{dfn}

With these definitions we formalize our observation that $Z$-type operators that commute with $S_1^X \cup S_2^X$ contain the former stabilizers or logical operators as substrings.
\begin{obs}
If $h\in N_3^Z$ then $\theta_1(h)\in N_1^Z$ and $\theta_2(h)\in N_2^Z$.
\end{obs}
\begin{proof}
By definition $h$ commutes with each element of $H_1^X$. The set $H_1^X$ has support only on qubits $Q_1$. Hence $\theta_1(h)\in N_1^Z$. A similar argument shows that $\theta_2(h)\in N_2^Z$.
\end{proof}

What this observation says is that if a $Z$-type operator commutes with all of the $X$-type operators, then it looks like a $Z$-type stabilizer or logical operator of each code when restricted to qubits of that code. With this motivation we define welding between two stabilizer codes of the CSS type. 
\begin{dfn}
Define 
\begin{equation}
H_3^X \equiv H_1^X \cup H_2^X
\end{equation}
\begin{equation}
S_3^Z \equiv \{ h \in \left< Z_1,...,Z_n \right>: \theta_1(h)\in S_1^Z, \theta_2(h)\in S_2^Z \}.
\end{equation}
 $S_3=\left< H_3^X \cup S_3^Z \right>$ is said to be a welded stabilizer group of the groups $S_1$ and $S_2$. They are welded on the qubits $Q_1\cap Q_2$. 
\end{dfn}

Not only do $Z$-type stabilizers get welded together, but also $Z$-type logical operators get welded together. 

\begin{obs}
If $l$ is a non-trivial $Z$-type logical operator, i.e. $l \in N_3^Z \backslash S_3^Z$, then either $\theta_1(l)\in N_1^Z\backslash S_1^Z$ or $\theta_2(l)\in N_2^Z \backslash S_2^Z$ or both.
\end{obs}
\begin{proof}
If $l\in N_3^Z\backslash S_3^Z$ then $l$ commutes with each element of the groups $S_1^X$ and $S_2^X$. Hence $\theta_1(l)\in N_1^Z$ and $\theta_2(l)\in N_2^Z$. If both $\theta_1(l) \in S_1^Z$ and $\theta_2(l) \in S_2^Z$ then $l\in S_3^Z$ which contradicts our assumption that it is not. Hence $\theta_1^Z(l)\in N_1^Z\backslash S_1^Z$ or $\theta_2(l) \in N_2^Z \backslash S_2^Z$ or both.
\end{proof}

Next, we define a special set of conditions under which nontrivial logical operators of the welded code exist and where the welded stabilizer group $S_3$ is local provided that the stabilizers groups $S_1$ and $S_2$ are local. Generally the welded code need not encode any qubits. That is, all logical operators might be proportional to the identity $I$. Also it is not guaranteed that two local codes necessarily weld into another local code. We now define conditions under which these problems go away.

\begin{dfn}
Two generating sets $H_1$ and $H_2$ are {\em well matched} when for all $h_1\in H_1$, there exists $h_2\in H_2$ such that $\theta_{12}(h_1)=\theta_{12}(h_2)$ and for all $h_2\in H_2$ there exists $h_1\in H_1$ such that $\theta_{12}(h_2)=\theta_{12}(h_1)$.
\end{dfn}

\begin{dfn}
Consider a set $W=\{ h\in H_i^Z : \theta_{12}(h)\neq I\}$. The generating set $H_i^Z$ is said to be {\em independent on the weld } if for all subsets $Y\subseteq W$ such that $Y\neq \phi$, the product $\prod_{h\in Y} h\neq I$. 
\end{dfn} 

\begin{thm}
\label{weldingtheorem}
If the sets $H_1^Z \cup L_1^Z$ and $H_2^Z\cup L_2^Z$ are well matched and independent on the weld and if the sets $L_1^Z$ and $L_2^Z$ are well matched and independent on the weld then:
\begin{enumerate}
\item A generating set for the $Z$-type stabilizers of the stabilizer code $S_3$ is given by $H_3^Z\equiv \{h_1 h_2 \theta_{12}(h_1): h_1 \in H_1^Z, h_2 \in H_2^Z, \theta_{12}(h_1)=\theta_{12}(h_2) \}$.
\item A generating set for the $Z$-type logical operators  of the stabilizer code $S_3$ is given by $L_3^Z\equiv \{l_1 l_2 \theta_{12}(l_1): l_1 \in L_1^Z, l_2 \in L_2^Z, \theta_{12}(l_1)=\theta_{12}(l_2) \}$.
\end{enumerate}
\end{thm}

Theorem \ref{weldingtheorem} tells us what the stabilizers group is and what the logical operators are. Hence we can find properties of the code such as the number of encoded qubits. For instance, suppose stabilizer groups $S_1$ and $S_2$ have $n_1$ and $n_2$ independent $Z$-type logical operators that have no support on qubits $Q_1\cap Q_2$ and both codes have $m$ $Z$-type logical operators that have support on $Q_1 \cap Q_2$.  The welded code has $n_1 + n_2 +m$ encoded qubits.  

To prove theorem \ref{weldingtheorem}, we first prove a lemma.
\begin{lem}
\label{gen1}
The set $H_3^Z \cup L_3^Z$ generates every element in the groups $\{ h\in N_2^Z : \theta_1(h)=I\}$ and $\{ h\in N_1^Z : \theta_2(h)=I\}$.
\end{lem}
\begin{proof}
We will show that if $g\in H_2^Z \cup L_2^Z$ such that $\theta_1(g)=I$, then $g\in H_3^Z \cup L_3^Z$ and that $H_2^Z \cup L_2^Z$ generates all of the group $\{ h\in N_2^Z : \theta_1(h)=I\}$. 

If $g\in H_2^Z$ such that $\theta_1(g)=I$ then we can find $h_1 \in H_1^Z$ and $h_2\in H_2^Z$ such that $\theta_{12}(h_1)=\theta_{12}(h_2)$ and $h_1 h_2 \theta_{12}(h_1)=g$. Namely when $h_1=g$ and $h_2=I$. Hence $g\in H_3^Z$. Similarly if $g\in L_2^Z$ such that $\theta_1(g)=I$ then we can find $l_1 \in L_1^Z$ and $l_2\in L_2^Z$ such that $\theta_{12}(l_1)=\theta_{12}(l_2)$ and $l_1 l_2 \theta_{12}(l_1)=g$. Namely when $l_1=g$ and $l_2=I$. Hence $g\in L_3^Z$.

The set $\{ H_2^Z \cup L_2^Z \}$ necessarily generates every element of the set $\{g \in N_2^Z: \theta_1(g)=I \}$. This is because no element of $\{ H_2^Z \cup L_2^Z \}$ with support on qubits $Q_1\cap Q_2$ can be included in the product by independence on the weld. 

Proving that the set $H_3^Z \cup L_3^Z$ generates every elements of the group  $\{ h\in N_1^Z : \theta_2(h)=I\}$ is symmetric to proving that it generates every element of the group $\{ h\in N_2^Z : \theta_1(h)=I\}$.
\end{proof}

We now prove theorem \ref{weldingtheorem}. 

\begin{proof}
We first prove item 1 that $H_3^Z$ generates $S_3^Z$. We have the identity $\theta_1(H_3^Z)=H_1^Z$ because 
\begin{align*}
&\theta_1(H_3^Z) \\
&=\{ \theta_1(h_1 h_2 \theta_{12}(h_1)):  h_1 \in H_1^Z, h_2 \in H_2^Z, \theta_{12}(h_1)=\theta_{12}(h_2) \} \\
&=H_1^Z.
\end{align*} Hence if $g \in S_3^Z$ then we can find $h\in \left< H_3^Z \right>$ such that $\theta_1(hg)=I$. The operator $hg$ is a stabilizer of $S_3$. By lemma \ref{gen1}, $hg$ must be generated by the set $H_3^Z$. Hence the operator $g$ is generated by the set $H_3^Z$. 

Next we prove item 2 that the set $L_3^Z$ generates all $Z$-type logical operators of the stabilizer group $S_3$. We need to prove that $\left< L_3^Z \cup H_3^Z \right> = N_3^Z$. We know that $\theta_1(\left< L_3^Z \cup H_3^Z \right>)=\theta_1(N_1^Z)$. Hence if $l\in N_3^Z$ then we can find $h\in \left<L_3^Z \cup H_3^Z \right>$ such that $\theta_1(hg)=I$. By lemma \ref{gen1}, $hg \in \left<L_3^Z \cup H_3^Z \right>$ and hence the operator $g$ is generated by $L_3^Z  \cup H_3^Z$. 
\end{proof}

Next we'll discuss locality. We'll show that local well-matched independent-on-the-weld CSS stabilizer groups weld into a local stabilizer group. First we'll need a notion of the width of a stabilizer.

\begin{dfn}
Define the distance between qubits $q_1,q_2 \in Q$ to be $d(q_1,q_2)$. Let the maximum width of an element in a set of Pauli operators $M$ be defined as $R(M)\equiv max \ d(q_1,q_2)$ such that $q_1,q_2 \in Q(\{ h \})$ for $h\in M$. 
\end{dfn}

\begin{cor}\
\label{distance}
If $H_1^Z$ and $H_2^Z$ are well matched and linearly independent on the weld then 
\begin{equation}
R(H_3^Z \cup H_3^X) \leq R(H_1^Z \cup H_1^X)+R(H_2^Z \cup H_2^X).
\end{equation}
\end{cor}

\begin{proof}
Let $h$ be an operator in $H_3^X \cup H_3^Z$ such that $R(\{h\})=R(H_3^X \cup H_3^Z)$.  If $h\in H_1^X$ then $R(H_3^Z\cup H_3^X)=R(\{ h \}) \leq R(H_1^Z\cup H_1^X) + R(H_2^Z\cup H_2^X)=R(\{h\})+R(H_2^Z\cup H_2^X)$. If $h\in H_2^X$ then the proof is the same as for $h\in H_1^X$ so without lack of generality we only prove one of those cases. If $h\in H_3^Z$ then $h=h_1 h_2 \theta_{12}(h_1)$ such that $h_1\in H_1^Z$, $h_2\in H_2^Z$ and $\theta_{12}(h_1)=\theta_{12}(h_2)$. By the triangle inequality $R(\{h_1 h_2 \theta_{12}(h_1) \})\leq R(\{ h_1 \}) + R(\{ h_2 \}) \leq R(H_1^Z\cup H_1^X) + R(H_2^Z\cup H_2^X)$.
\end{proof}

Corollary \ref{distance}, along with theorem \ref{weldingtheorem}, allow us to combine the shapes of the logical operators from two codes while keeping the generating set for the resulting code local. One can even design local codes with large energy barriers using this technique as demonstrated with the welded solid code. 

To summarize, combining two CSS stabilizer codes so that the stabilizers commute lead to the idea that the $Z$-type stabilizer should be updated so as to achieve this. When the $Z$-type stabilizers are updated, they become welded versions of the $Z$-type stabilizers from the former codes.  If the codes are well matched and independent on the weld, then we can weld generating sets of stabilizers and logical operators.  If the former generating sets are local then so will be the welded generating set.

\end{document}